\documentclass[twocolumn,pra,showpacs,superscriptaddress]{revtex4-1}
\usepackage{tikz}

\usepackage{graphicx}
\usepackage{mathtools}
\usepackage{enumerate,amsmath,amssymb,amstext,amsthm}
\usepackage{mathrsfs}
\usepackage{verbatim}

\newtheorem{thm}{Theorem}

\newtheorem{lmm}[thm]{Lemma}

\def\be{\begin{equation}}
\def\ee{\end{equation}}
\def\bi{\begin{itemize}}
\def\ei{\end{itemize}}

\newcommand{\R}{\mathbb{R}}

\newcommand{\cor}{\mathrm{Cor}}
\newcommand{\beh}{\mathrm{Beh}}
\newcommand{\cut}{\mathrm{Cut}}
\newcommand{\met}{\mathrm{Met}}

\newcommand{\calE}{\mathcal{E}}

\def\bi{\begin{itemize}}
\def\ei{\end{itemize}}

\newcommand{\abs}[1]{| #1 |}

\newcommand{\tr}{\mathrm{tr}}

\newcommand{\ip}[2]{\left\langle #1 , #2\right\rangle} 

\begin{document}

\title{Computing quantum Bell inequalities}


\author{Le Phuc Thinh}
\email{cqtlept@nus.edu.sg}
\affiliation{Centre for Quantum Technologies, National University of Singapore}


\begin{abstract}
Understanding the limits of quantum theory in terms of uncertainty and correlation has always been a topic of foundational interest. Surprisingly this pursuit can also bear interesting applications such as device-independent quantum cryptography and tomography or self-testing. Building upon a series of recent works on the geometry of quantum correlations, we are interested in the problem of computing quantum Bell inequalities or the boundary between quantum and post-quantum world. Better knowledge of this boundary will lead to more efficient device-independent quantum processing protocols. We show that computing quantum Bell inequalities is an instance of a quantifier elimination problem, and apply these techniques to the bipartite scenario in which each party can have three measurement settings. Due to heavy computational complexity, we are able to obtain the characterization of certain ``linear" relaxation of the quantum set for this scenario. The resulting quantum Bell inequalities are shown to be equivalent to the Tsirelson-Landau-Masanes arcsin inequality, which is the only type of quantum Bell inequality found since 1987.
\end{abstract}

\maketitle

\section{\label{sec:intro}Introduction}
One of the most striking feature of quantum theory is the phenomenon of Bell nonlocality: the theory allows certain special correlations between distant observers that cannot be explained by any local hidden variable model~\cite{bell,bookNL,reviewNL}. Such correlations form the basis of cryptographic security, as well as many modern applications of quantum phenomena.

The set of correlations admitting local hidden variable description has been shown to be a convex polytope whose extreme points can be efficiently enumerated for any Bell scenario~\cite{Fine}. This has important consequence for experiment: it is a linear programming feasibility problem to test if a behavior is local or nonlocal. However, it is an entirely different problem to ask for {\em all the Bell inequalities}, i.e. facets of the local polytope, which remains a very difficult open problem in quantum information~\cite{allBell}. The main difficulty is related to the computational complexity of this task, and perhaps to our lack of knowledge about an underlying principles relating the various Bell inequalities. 

What is the analogous story for quantum correlations? In the simplest Bell scenario, Tsirelson, independently Landau and Masanes, derived the nonlinear ``quantum Bell inequality"
$$\arcsin(E_{11})+\arcsin(E_{12})+\arcsin(E_{21})-\arcsin(E_{22})\leq\pi$$
that must be obeyed by any quantum correlation~\cite{TS87,Landau,Masanes}. Later, Navascues {\em et al.} proposed a convergence hierarchy of semidefinite programs for testing quantum behaviors~\cite{npa}. This is clearly an analogue of the linear programming feasibility problem for membership within the local polytope, although the hierarchy quickly grows in size making the method impractical. The problem of computing {\em all the quantum Bell inequalities} for a given scenario can quickly be recognized, but one immediately run into serious technical issues.

With experimental progress, perhaps one day we could attempt testing the limit of quantum theory, particularly the limit of quantum correlations. Thus, knowledge about the boundary between quantum and post-quantum correlations is likely to assist designing new experiments near the boundary. This requires a better understanding of the quantum set of correlation, particularly its boundary structure~\cite{geometry}. Additionally, knowledge of the quantum set can be useful in designing new tomographic methods similar to self-testing~\cite{MY} and gate set tomography~\cite{GST}, as well as optimal device-independent protocols for information processing.

Progress in deriving quantum Bell inequalities is very rare. If one is only concern with necessary conditions, i.e. bounding the quantum set, then a flexible method based on the principle of macroscopic locality can be employed~\cite{qBell_ML}. However, if one cares about both necessary and sufficient conditions, i.e. quantum Bell inequalities corresponding to maximal faces, then the only significant progress known to date is made by Tsirelson for the correlation scenario~\cite{TS87}.

In this work we attempt to understand better the process of deriving or computing quantum Bell inequalities. In particular, we are interested in deriving a new type of quantum Bell inequality. We make the observation that whenever the NPA hierarchy collapse to a finite level, as is the case for correlation scenarios, computing quantum Bell inequalities can be thought of as a {\em quantifier elimination} problem. We give several characterizations of $\cor(3,3)$, the quantum set of bipartite correlation with three measurement choices per party. We found that generic quantifier elimination algorithms are unable to handle this scenario. Fortunately, we can perform quantifier elimination for a ``linear" relaxation of $\cor(3,3)$, and found several apparent new quantum Bell inequalities, which turns out to be implied by the system of Tsirelson's inequalities. Lastly, we sketch several connections between $\cor(n,m)$ with the metric and the cut polytopes, which are useful in understanding quantum correlation sets in general.

\section{\label{sec:nonlocality}Quantum correlations}
Quantum correlations arises from experiments of the following type: there are two space-like separated parties, called Alice and Bob, perform independently and simultaneously local {\em dichotomic} measurements on their corresponding subsystems of a joint quantum system, and record the resulting outcomes. The event that the Alice performed measurement $x$ and got outcome $a$, and Bob performed measurement  $y$ and got outcome $b$ is denoted by $(a,b|x,y)$. By repetition, they can estimate a family of probabilities $p(a,b|x,y)$, known as the \emph{full behavior}, for all possible events. The full behavior characterizes the probabilistic response of the quantum system under chosen measurements and can be thought of as a shadow of the quantum state.

A full behavior $p(a,b|x,y)$ coming from quantum theory must be a result of measuring a quantum state with some choices of quantum measurement. That is there exists state $\rho$, and local measurement POVMs $\{M^x_a\}$ and $\{M^y_b\}$ such that
$$p(a,b|x,y)=\tr(\rho\,M^x_a\otimes M^y_b)\,.$$
Quantum behaviors are resources that cannot be used for superluminal communication because no-signalling holds
\begin{equation}
\label{eq:no-sig}
\begin{aligned}
 \sum_b p(a,b|x,y) &= \sum_b p(a,b|x,y') =: p_A(a|x)\,,\\
\sum_a p(a,b|x,y) &= \sum_a p(a,b|x',y) =: p_B(b|y)\,.
\end{aligned}
\end{equation}

Dichotomic measurement means that the measurement has two possible outcomes, which we will denote by $+1$ and $-1$. This allows the use of correlators to describe equivalently a behavior: $p(a,b|x,y) \leftrightarrow (c_x, c_y, c_{xy})$
where
$$c_x = \sum_{a\in \{\pm 1\}}a\,p_A(a|x) \quad 
c_y = \sum_{b\in \{\pm 1\}} b\,p_B(b|y),$$ 
$$c_{xy} = \sum_{a,b\in \{\pm 1\}} ab\,p(a,b|x,y).$$
Under this parametrization, a {\em full correlator}  $(c_x,c_y, c_{xy})$ is quantum if
\be\label{qcorrelation}
\begin{aligned}
c_x&=\tr(\rho(A_x\otimes I))\,,\\
c_y & =\tr(\rho(I\otimes B_y))\,,\\
 c_{xy} & = \tr(\rho(A_x\otimes B_y))\,,
\end{aligned}
\ee
for some choice of $\pm 1$ {\em observables} $A_x, B_y$ (via Naimark dilation).

The set of full behaviors or full correlators is denoted $\beh(n,m)$, whereas the projection of this set along joint correlation coordinates $c_{xy}$ is denoted as $\cor(n,m)$. The lack of {\em local marginals} turns out to have deep consequences. Even for the simplest scenario $\beh(2,2)$ can only be approximated using the NPA hiearchy, whereas $\cor(n,m)$ admits {\em exact} description for arbitrary $n,m$. In other words, it is not known that the NPA hiearchy for $\beh(2,2)$ collapse to any finite level.  

{\bf Important notation:} The set of quantum correlators can also be described by a nonlinear coordinate system: the pairwise-angle coordinates. The coordinate transformation is defined as $c \mapsto \hat{c}:=\arccos(c)\in[0,\pi]$ individually for each pair of setting $(x,y)$. Thus, a point of $\cor(n,m)$ can be equivalently specified as $(\hat{c}_{xy})$ in the pairwise-angle parametrization and corresponds to the point $(\cos\hat{c}_{xy})$ in the usual linear coordinates. We use this notation throughout the paper.

\section{\label{sec:cor33}Exact characterizations of $\cor(3,3)$}

Exact analytic description of $\cor(2,m)$ has been obtained in a previous work~\cite{Thinh}: $(c_{xy})\in \cor(2,m)$ if and only if it satisfies the following system
\begin{gather*}
0 \leq \hat{c}_{xy}\leq \pi,\\
0 \leq \hat{c}_{1j} + \hat{c}_{2i} + \hat{c}_{2j} - \hat{c}_{1i} \leq 2\pi ,\\
0 \leq \hat{c}_{1i} + \hat{c}_{2i} + \hat{c}_{2j} - \hat{c}_{1j} \leq 2\pi, \\
0 \leq \hat{c}_{1i} + \hat{c}_{1j} + \hat{c}_{2j} - \hat{c}_{2i} \leq 2\pi, \\
0 \leq \hat{c}_{1i} + \hat{c}_{1j} + \hat{c}_{2i} - \hat{c}_{2j} \leq 2\pi,
\end{gather*}
for all $x\in[2],y\in[m]$ and $3\le i<j\le m+2$. The trivial inequalities $0 \leq \hat{c}_{xy}\leq \pi$ or equivalently $-1 \leq c_{xy}\leq 1$ are called the {\em box inequalities}. The resulting inequalities are new in the sense that $\cor(2,n)$ has never been characterized, but only of the Tsirelson-Landau-Masanes {\em type}. Therefore the next scenario which we could hope to find new quantum Bell inequalities and moreover new type of inequalities is the $(3,3)$ scenario where Alice and Bob each has three dichotomic measurements. In this section, we study various characterizations of the set of quantum correlation for this scenario.

We start with Tsirelson's characterization of $\cor(3,3)$: $(c_{xy})\in\cor(3,3)$ if and only if there exists six unit vectors $u_x,v_y\in\R^{6}$ such that $\ip{u_x}{v_y}=c_{xy}$. This characterization of the set of all possible quantum correlations looks deceptively simple, but is in fact a remarkable consequence of correlation scenarios being describable by a Clifford algebra. The geometric interpretation of Tsirelson's result gives rise to the pairwise-angle parametrization which we mentioned in Section~\ref{sec:nonlocality}.

Tsirelson's characterization, however, does not give an {\em explicit solution} to the question of when a given correlation is quantum. For this purpose, we can rephrase Tsirelson' result in the following matrix-theoretic form: $(c_{xy})\in\cor(3,3)$ if the following partial symmetric matrix admits a psd completion
$$
\begin{pmatrix}
1 & ? &  ? &  c_{11} & c_{12} & c_{13}\\
? & 1  & ? & c_{21} & c_{22}& c_{23} \\
? & ? & 1& c_{31} & c_{32}& c_{33} \\
c_{11} & c_{21} & c_{31} &  1& ?& ? \\
c_{12} & c_{22} & c_{32} &  ? & 1 & ? \\
c_{13} & c_{23}  & c_{33} & ? & ? & 1  
\end{pmatrix}\,.
$$
This form is convenient for numerical investigations due to the fact that semidefinite programming can be solved efficiently. Further, this form of characterizing quantum correlations brings new perspectives and tools---matrix completion, geometry of the elliptope---for thinking and solving problems involving quantum correlations~\cite{wolk,L97,BJT,shin}. For example, one can use this form to find optimal quantum violation of Bell correlation inequalities.

The psd completion characterization of $\cor(3,3)$, though convenient for numerical works, can in fact be used for analytic investigations although perhaps limited to small size problems. {\em Sylvester's criterion} for positive-semidefiniteness states that a Hermitian matrix is psd if and only if all principal minors of the matrix are nonnegative~\cite{matrix_analysis}. Recall that a principal submatrix is a smaller matrix obtained from the original matrix by selecting (equivalently deleting) the same index set of rows and columns; a principal minor is the determinant of a principal submatrix.

By Sylvester's criterion, $(c_{xy})\in\cor(3,3)$ if and only if there exists real numbers $\alpha,\beta,\gamma,\delta,\chi,\xi$ such that all $2^6-1$ principal minors of
$$
\begin{pmatrix}
1 & \alpha &  \beta &  c_{11} & c_{12} & c_{13}\\
? & 1  & \gamma & c_{21} & c_{22}& c_{23} \\
? & ? & 1& c_{31} & c_{32}& c_{33} \\
c_{11} & c_{21} & c_{31} &  1& \delta & \chi \\
c_{12} & c_{22} & c_{32} &  ? & 1 & \xi \\
c_{13} & c_{23}  & c_{33} & ? & ? & 1  
\end{pmatrix}\,
$$
are non-negative. However, this naive approach leads to an untractable solution because of the {\em exponential} growth in the number of principal minors. Using the results related to {\em chordal completion}~\cite{BJT}, one can simplify the situation somewhat and obtain a more efficient characterization by reducing the number of principal minors. This is our first result:
\begin{lmm}\label{pro:cor33analytic}
$(c_{xy})\in\cor(3,3)$ if and only if there exists $\alpha,\beta,\gamma\in[-1,1]$ satisfying the system of polynomial inequalities for $y=1,2,3$
\begin{equation}\label{eq:cor33_linear}
\begin{aligned}
1&-\sum_{x=1}^3 c_{xy}^2 -\alpha^2 -\beta^2 -\gamma^2 +2c_{1y}c_{2y}\alpha +2c_{1y}c_{3y}\beta \\
  &+ 2c_{2y}c_{3y}\gamma +2\alpha\beta\gamma +c_{3y}^2\alpha^2 +c_{2y}^2\beta^2 +c_{1y}^2\gamma^2 \\
  &-2c_{2y}c_{3y}\alpha\beta -2c_{1y}c_{3y}\alpha\gamma - 2c_{1y}c_{2y}\beta\gamma \geq0\,, \\
1&-c_{2y}^2 -c_{3y}^2 -\gamma^2 +2c_{2y}c_{3y}\gamma \geq0\,, \\
1&-c_{1y}^2 -c_{3y}^2 -\beta^2 +2c_{1y}c_{3y}\beta \geq0\,, \\
1&-c_{1y}^2 -c_{2y}^2 -\alpha^2 +2c_{1y}c_{2y} \geq0\,, \\
\end{aligned}
\end{equation}
or equivalently if and only if there exists $\hat\alpha,\hat\beta,\hat\gamma\in[0,\pi]$ satisfying the linear system of inequalities for $y=1,2,3$
\begin{equation}\label{eq:cor33analytic_angles_linear}
\begin{aligned}
&\hat{\alpha}\leq\hat{\beta}+\hat{\gamma}, \hat{\beta}\leq\hat{\alpha}+\hat{\gamma}, \hat{\gamma}\leq\hat{\alpha}+\hat{\beta}, \hat{\alpha}+\hat{\gamma}+\hat{\beta}\leq2\pi,\\
&\hat{\alpha}\leq\hat{c}_{1y}+\hat{c}_{2y}, \hat{c}_{1y}\leq\hat{c}_{2y}+\hat{\alpha}, \hat{c}_{2y}\leq\hat{\alpha}+\hat{c}_{1y},\\
&\hat{\alpha}+\hat{c}_{1y}+\hat{c}_{2y}\leq2\pi,\\
&\hat{\beta}\leq\hat{c}_{1y}+\hat{c}_{3y}, \hat{c}_{1y}\leq\hat{c}_{3y}+\hat{\beta},  \hat{c}_{3y}\leq\hat{\beta}+\hat{c}_{1y}\\
&\hat{\beta}+\hat{c}_{1y}+\hat{c}_{3y}\leq2\pi, \\
&\hat{\gamma}\leq\hat{c}_{2y}+\hat{c}_{3y}, \hat{c}_{2y}\leq\hat{c}_{3y}+\hat{\gamma}, \hat{c}_{3y}\leq\hat{\gamma}+\hat{c}_{2y}\\
&\hat{\gamma}+\hat{c}_{2y}+\hat{c}_{3y}\leq2\pi,
\end{aligned}
\end{equation}
and the nonlinear system of inequalities for $y=1,2,3$
\begin{equation}\label{eq:cor33analytic_angles_nonlinear}
\begin{aligned}
(&\sin\hat{c}_{1y}\sin\hat{c}_{2y}\sin\hat{c}_{3y})^2\, \\
&- (\sin\hat{c}_{3y}(\cos\hat{\alpha}-\cos\hat{c}_{1y}\cos\hat{c}_{2y}))^2\, \\
&- (\sin\hat{c}_{2y}(\cos\hat{\beta}-\cos\hat{c}_{1y}\cos\hat{c}_{3y}))^2\, \\
& - (\sin\hat{c}_{1y}(\cos\hat{\gamma}-\cos\hat{c}_{2y}\cos\hat{c}_{3y}))^2\, \\ 
& - 2(\cos\hat{\alpha}-\cos\hat{c}_{1y}\cos\hat{c}_{2y})(\cos\hat{\beta}-\cos\hat{c}_{1y}\cos\hat{c}_{3y})\, \\
&\hphantom{-2(\cos\hat{\alpha}-\cos\hat{c}_{1y}\cos)}(\cos\hat{\gamma}-\cos\hat{c}_{2y}\cos\hat{c}_{3y}) \geq 0\,.
\end{aligned}
\end{equation}
\end{lmm}
\begin{proof}
A necessary condition on the correlations is that the partial symmetric matrix must be partial positive-semidefinite. This means that any specified principal submatrix (consisting of known entries) is psd. This leads to box inequalities on the correlations $(c_{xy})$, i.e. $-1\leq c_{xy}\leq1$ for all $x,y$. If we can find entries $\alpha,\beta,\gamma$ such that
$$
X = \begin{pmatrix}
1 &\alpha &  \beta &  c_{11} & c_{12} & c_{13}\\
\alpha & 1  & \gamma& c_{21} & c_{22}& c_{23} \\
\beta & \gamma & 1& c_{31} & c_{32}& c_{33} \\
c_{11} & c_{21} & c_{31} &  1& ?& ? \\
c_{12} & c_{22} & c_{32} &  ? & 1 & ? \\
c_{13} & c_{23}  & c_{33} & ? & ? & 1  
\end{pmatrix}
$$
is partial positive-semidefinite, then by chordal completion~\cite{BJT}, the missing $?$ entries can be filled making the whole matrix psd. (This is because the graph associated with the matrix completion problem of $X$ is now a chordal graph, i.e. any circuit with length at least 4 has a chord, and sufficient conditions for completion to exist are known.) Hence, $(c_{xy})$ satisfying the box inequalities is a member of $\cor(3,3)$ if and only if there exist $\alpha,\beta,\gamma$ such that $X$ is partial psd.

Now the condition that $X$ is partial psd means for any specified principal submatrix is psd. It suffices to consider {\em maximal} specified principal submatrices
$$
X_y = \begin{pmatrix}
1 &\alpha &  \beta &  c_{1y} \\
\alpha & 1  & \gamma& c_{2y} \\
\beta & \gamma & 1& c_{3y} \\
c_{1y} & c_{2y} & c_{3y} & 1\\
\end{pmatrix} \text{ for } y=1,2,3\,.
$$
Viewing $X_y$ as the block matrix
$$
\left(\begin{matrix}
A & B_y\\
C_y & D
\end{matrix}\right)
$$
where $D = 1$ and 
$$
A = \left(\begin{matrix}
1 &\alpha &  \beta \\
\alpha & 1  & \gamma\\
\beta & \gamma & 1
\end{matrix}\right)\,,
B_y = \left(\begin{matrix}
c_{1y}\\
c_{2y}\\
c_{3y}
\end{matrix}\right), C_y = B_y^\intercal,
$$
and applying Schur's complement for positive-semidefiniteness~\cite{matrix_analysis}:
$$
X_y\geq0 \Leftrightarrow A-B_yB_y^\intercal\geq0
$$
we have $c_{xy}\in\cor(3,3)$ iff there exist $\alpha,\beta,\gamma$ that for $y=1,2,3$,
\begin{align}
\begin{pmatrix}
1 &\alpha &  \beta \\
\alpha & 1  & \gamma\\
\beta & \gamma & 1
\end{pmatrix} - \begin{pmatrix}
c_{1y}^2 & c_{1y}c_{2y} & c_{1y}c_{3y}\\
c_{1y}c_{2y} & c_{2y}^2 & c_{2y}c_{3y}\\
c_{1y}c_{3y} & c_{2y}c_{3y} & c_{3y}^2
\end{pmatrix}\geq0
\end{align}
Using Sylvester's criterion for each $y$ result in the first characterization.

To obtain the second characterization in the pairwise-angle parametrization, we use the well-known linearization of the elliptope of dimension three~\cite{BJT}: for $\theta_1,\theta_2,\theta_3\in[0,\pi]$, the matrix 
$$
\begin{pmatrix}1 & \cos\theta_1& \cos\theta_2\\
\cos\theta_1& 1 & \cos\theta_3\\
\cos\theta_2 & \cos\theta_3 & 1 
\end{pmatrix}$$
is positive semidefinite if and only if 
\begin{gather*}
\theta_1\leq\theta_2+\theta_3, \quad \theta_2\leq\theta_3+\theta_1, \quad \theta_3\leq\theta_1+\theta_2,\\
\theta_1 + \theta_2 + \theta_3 \le 2\pi.
\end{gather*}
The condition that $X_y\geq0$ for $y=1,2,3$ is equivalent by Sylvester's criterion to the system of linear inequalities in the angles except for the set of inequalities arising from the 4-by-4 minors. By Schur's determinant identity, these inequalities are
$$
\text{det}(X_y) = \text{det}\left[\left(\begin{matrix}
1 &\cos\hat{\alpha} &  \cos\hat{\beta} \\
\cos\hat{\alpha} & 1  & \cos\hat{\gamma}\\
 \cos\hat{\beta} & \cos\hat{\gamma} & 1
\end{matrix}\right)-B_yB_y^\intercal\right]\geq0\,,
$$
now re-expressed in the pairwise-angle. This completes the characterization of $\cor(3,3)$ in pairwise-angle coordinates.
\end{proof}

The characterization~\eqref{eq:cor33_linear} is exact but still implicit due to the existence of the unknown variables $\alpha,\beta,\gamma$. To obtain a complete solution, analogous to the description of $\cor(2,n)$, we have to eliminate the variables $\alpha,\beta,\gamma$ from the system of Proposition~\ref{pro:cor33analytic}. Such objectives can be done in principle using the method of {\em quantifier elimination}~\cite{QECAD}, which is a concept of simplification originating from mathematical logic, model theory and computer science, but has since found applications in various fields of science and engineering. The most familiar example of quantifier elimination is the high school algebra fact
\begin{align*}
\exists x\in\mathbb{R},(a\neq0 &\land ax^2+bx+c=0)\\
&\iff (a\neq0 \land b^2-4ac\geq0)\,.
\end{align*}
The right hand side is an equivalent quantifier free formula for the left hand side, which answers the question ``when does a single variable quadratic polynomial equation has a real root?"

We implemented a program to compute quantum Bell inequalities for $\cor(3,3)$ in QEPCAD~\cite{QEPCAD} as well as REDLOG~\cite{REDLOG}. Unfortunately, our problem is a ``large-size" problem for generic quantifier elimination algorithms to handle. Moreover, these algorithms, as far as we understand, do exploit additional symmetry of $\cor(3,3)$, nor the elliptope geometry induced by the convex geometry of the psd cone. Consequently, we are unable to obtain a definitive characterization, i.e. a quantifier free description, of $\cor(3,3)$ purely in terms of the given correlations $(c_{xy})$.

\section{\label{sec:relaxation}Relaxation of $\cor(3,3)$ and quantum Bell inequalities}
Apart from the inequalities from the 4-by-4 minors, the system of inequalities~\eqref{eq:cor33analytic_angles_linear} are all linear in the pairwise-angle coordinate. This is an important fact that greatly simplifies the process of quantifier elimination. Therefore, if we settle for a bound on the quantum set $\cor(3,3)$, a good candidate will be the set defined by those linear inequalities. We begin with a characterization of the relaxed set.
\begin{lmm}\label{lmm:relax_equiv}
There exists $\hat\alpha,\hat\beta,\hat\gamma\in[0,\pi]$ satisfying the linear system of inequalities~\eqref{eq:cor33analytic_angles_linear}
if and only if the correlators satisfy the linear system for all $x,x',y,y',{\bar y}\in\{1,2,3\}$
\begin{gather*}
\abs{\hat{c}_{xy}-\hat{c}_{x'y}} + \abs{\hat{c}_{xy'} + \hat{c}_{x'y'} - \pi} \leq \pi\,, \\
\abs{\hat{c}_{1y}-\hat{c}_{2y}} + \abs{\hat{c}_{2y'}-\hat{c}_{3y'}} + \abs{\hat{c}_{1\bar y} -\hat{c}_{3\bar y}} \leq 2\pi\,,  \\
\abs{\hat{c}_{1y}-\hat{c}_{2y}} + \abs{\hat{c}_{2\tilde y'} + \hat{c}_{3\tilde y'} - \pi} + \abs{\hat{c}_{1\bar y} + \hat{c}_{3\bar y} - \pi}  \leq 2\pi\,, \\
\abs{\hat{c}_{1y} + \hat{c}_{2y} - \pi} + \abs{\hat{c}_{2y'}-\hat{c}_{3y'}} + \abs{\hat{c}_{1\bar y} + \hat{c}_{3\bar y} - \pi}  \leq 2\pi\,,  \\
\abs{\hat{c}_{1y} + \hat{c}_{2y} - \pi} + \abs{\hat{c}_{2y'} + \hat{c}_{3y'} - \pi} + \abs{\hat{c}_{1\bar y}-\hat{c}_{3\bar y}} \leq 2\pi\,.
\end{gather*}
\end{lmm}
\begin{proof}
Our strategy is to gather all the inequalities involving $\hat{\alpha}$, eliminate this variable by enforcing any lower bound must be less than or equal to any upper bounds, then repeat for $\hat{\beta}$ and $\hat{\gamma}$. For compactness of the proof we rewrite
$$
\begin{rcases}
\abs{\beta-\gamma}
\end{rcases}
\leq\alpha\leq
\begin{cases}
\beta + \gamma\\
2\pi - (\beta + \gamma)
\end{cases}
$$
as a single compound inequality
$$
\abs{\beta-\gamma} \leq \alpha \leq \pi - \abs{\beta + \gamma - \pi}\,.
$$

First, the $\hat{\alpha}$-system of inequalities is given by
\begin{equation*}
\begin{rcases}
\abs{\hat{\beta}-\hat{\gamma}}\\
\abs{\hat{c}_{1y}-\hat{c}_{2y}}
\end{rcases}
\leq \hat{\alpha} \leq
\begin{cases}
\pi - \abs{\hat{\beta} + \hat{\gamma} - \pi}\\
\pi - \abs{\hat{c}_{1y'} + \hat{c}_{2y'} - \pi}
\end{cases}
\end{equation*}
for $y,y'\in\{1,2,3\}$, after we have absorb the inequalities $0\leq\hat{\alpha}\leq\pi$ into existing inequalities (due to $\abs{\hat{c}_{1y}-\hat{c}_{2y}}\geq0$ and $\pi - \abs{\hat{c}_{1y'} + \hat{c}_{2y'} - \pi}\leq\pi$). Eliminating $\hat{\alpha}$ gives
\begin{align}
\abs{\hat{\beta}-\hat{\gamma}} &\leq \pi - \abs{\hat{\beta} + \hat{\gamma} - \pi} \nonumber \\
\abs{\hat{\beta}-\hat{\gamma}} &\leq \pi - \abs{\hat{c}_{1y'} + \hat{c}_{2y'} - \pi} \nonumber \\
\abs{\hat{c}_{1y}-\hat{c}_{2y}} &\leq \pi - \abs{\hat{\beta} + \hat{\gamma} - \pi} \nonumber \\
\abs{\hat{c}_{1y}-\hat{c}_{2y}} &\leq \pi - \abs{\hat{c}_{1y'} + \hat{c}_{2y'} - \pi} \label{eq:12}
\end{align}
for all $y,y'\in\{1,2,3\}$. The numbered inequality does not involve the variable $\hat{\beta},\hat{\gamma}$ and is not required in the next reduction.

Second, the remaining system of inequalities is equivalent to
\begin{gather*}
0 \leq \hat{\beta} \leq \pi \text{ and } 0 \leq \hat{\gamma} \leq \pi\,, \\
\hat{\gamma} -\pi + \abs{\hat{c}_{1y'} + \hat{c}_{2y'} - \pi} \leq \hat{\beta} \leq \hat{\gamma} + \pi - \abs{\hat{c}_{1y'} + \hat{c}_{2y'} - \pi}\,,\\
-\hat{\gamma} + \abs{\hat{c}_{1y}-\hat{c}_{2y}}\leq \hat{\beta} \leq 2\pi -\hat{\gamma} - \abs{\hat{c}_{1y}-\hat{c}_{2y}}\,,
\end{gather*}
from which the $\hat{\beta}$-system of inequalities follows
\begin{align*}
\begin{rcases}
\abs{\hat{c}_{1\bar y}-\hat{c}_{3\bar y}}\\
\hat{\gamma} -\pi + \abs{\hat{c}_{1y'_1} + \hat{c}_{2y'_1} - \pi}\\
-\hat{\gamma} + \abs{\hat{c}_{1y_1}-\hat{c}_{2y_1}}
\end{rcases}
&\leq \hat{\beta} \\
\text{ and } \hat{\beta} \leq &
\begin{cases}
\pi - \abs{\hat{c}_{1\bar y'} + \hat{c}_{3\bar y'} - \pi}\\
\pi + \hat{\gamma} - \abs{\hat{c}_{1y'} + \hat{c}_{2y'} - \pi}\\
2\pi -\hat{\gamma} - \abs{\hat{c}_{1y}-\hat{c}_{2y}}
\end{cases}
\end{align*}
after absorbing $0\leq\hat{\beta}\leq\pi$.
Eliminating $\hat{\beta}$ gives
\begin{align}
\abs{\hat{c}_{1\bar y}-\hat{c}_{3\bar y}} &\leq \pi - \abs{\hat{c}_{1\bar y'} + \hat{c}_{3\bar y'} - \pi} \label{eq:13} \\
\abs{\hat{c}_{1\bar y}-\hat{c}_{3\bar y}} &\leq \pi +\hat{\gamma} - \abs{\hat{c}_{1y'} + \hat{c}_{2y'} - \pi} \nonumber \\
\abs{\hat{c}_{1\bar y}-\hat{c}_{3\bar y}} &\leq 2\pi -\hat{\gamma} - \abs{\hat{c}_{1y}-\hat{c}_{2y}} \nonumber \\
\hat{\gamma} -\pi + \abs{\hat{c}_{1y'_1} + \hat{c}_{2y'_1} - \pi} &\leq  \pi - \abs{\hat{c}_{1\bar y'} + \hat{c}_{3\bar y'} - \pi} \nonumber \\
\hat{\gamma} -\pi + \abs{\hat{c}_{1y'_1} + \hat{c}_{2y'_1} - \pi} &\leq \pi + \hat{\gamma} - \abs{\hat{c}_{1y'} + \hat{c}_{2y'} - \pi} \label{eq:trivial1}\\
\hat{\gamma} -\pi + \abs{\hat{c}_{1y'_1} + \hat{c}_{2y'_1} - \pi} &\leq 2\pi -\hat{\gamma} - \abs{\hat{c}_{1y}-\hat{c}_{2y}} \nonumber \\
-\hat{\gamma} + \abs{\hat{c}_{1y_1}-\hat{c}_{2y_1}} &\leq  \pi - \abs{\hat{c}_{1\bar y'} + \hat{c}_{3\bar y'} - \pi} \nonumber \\
-\hat{\gamma} + \abs{\hat{c}_{1y_1}-\hat{c}_{2y_1}} &\leq \pi +\hat{\gamma} - \abs{\hat{c}_{1y'} + \hat{c}_{2y'} - \pi} \nonumber \\
-\hat{\gamma} + \abs{\hat{c}_{1y_1}-\hat{c}_{2y_1}} &\leq 2\pi -\hat{\gamma} - \abs{\hat{c}_{1y}-\hat{c}_{2y}} \label{eq:trivial2}
\end{align}

Third, the $\hat{\gamma}$-system of inequalities is given by
\begin{align*}
\begin{rcases}
\abs{\hat{c}_{2\tilde y}-\hat{c}_{3\tilde y}}\\
-\pi + \abs{\hat{c}_{1\bar y}-\hat{c}_{3\bar y}} + \abs{\hat{c}_{1y'} + \hat{c}_{2y'} - \pi} \\
-\pi + \abs{\hat{c}_{1\bar y'} + \hat{c}_{3\bar y'} - \pi} + \abs{\hat{c}_{1y_1}-\hat{c}_{2y_1}}  \\
- \pi/2  + \abs{\hat{c}_{1y_1}-\hat{c}_{2y_1}}/2 + \abs{\hat{c}_{1y'} + \hat{c}_{2y'} - \pi}/2 
\end{rcases}
\leq \hat{\gamma} \\
\text{ and } \hat{\gamma} \leq
\begin{cases}
\pi - \abs{\hat{c}_{2\tilde y'} + \hat{c}_{3\tilde y'} - \pi}\\
2\pi - \abs{\hat{c}_{1y}-\hat{c}_{2y}} - \abs{\hat{c}_{1\bar y}-\hat{c}_{3\bar y}}\\
2\pi - \abs{\hat{c}_{1\bar y'} + \hat{c}_{3\bar y'} - \pi} - \abs{\hat{c}_{1y'_1} + \hat{c}_{2y'_1} - \pi}  \\
3\pi/2 - \abs{\hat{c}_{1y}-\hat{c}_{2y}}/2 - \abs{\hat{c}_{1y'_1} + \hat{c}_{2y'_1} - \pi}/2 \\
\end{cases}
\end{align*}
from which we can readily eliminate $\hat{\gamma}$.

Finally, gathering the numbered inequalities together with the inequalities from eliminating $\hat{\gamma}$ gives a system of inequality equivalent to the starting system. We proceed with further analysis and simplification. The inequalities~\eqref{eq:12} and~\eqref{eq:13} and
$$
\abs{\hat{c}_{2\tilde y}-\hat{c}_{3\tilde y}} \leq \pi - \abs{\hat{c}_{2\tilde y'} + \hat{c}_{3\tilde y'} - \pi}
$$
from the first pair of bounds in the $\hat{\gamma}$-system are combined into (for $x\neq x'$ originally and then for any $x,x'$)
\begin{equation}
\abs{\hat{c}_{xy}-\hat{c}_{x'y}} + \abs{\hat{c}_{xy'} + \hat{c}_{x'y'} - \pi}\leq \pi. \label{eq:123}
\end{equation}
Observe that for $x\neq x'$ and $y=y'$ (or for $x=x'$ and any $y,y'$), inequalities~\eqref{eq:123} reduces to box inequalities $\hat{c}_{xy}\in[0,\pi]$ for all $x,y\in\{1,2,3\}$. Next, box inequalities together imply~\eqref{eq:trivial1} and~\eqref{eq:trivial2} so it suffices to keep the box inequalities. Likewise, box inequalities together with~\eqref{eq:123} imply all of the inequalities from the $\hat{\gamma}$-system except
\begin{gather*}
\abs{\hat{c}_{2\tilde y}-\hat{c}_{3\tilde y}} \leq 2\pi - \abs{\hat{c}_{1y}-\hat{c}_{2y}} - \abs{\hat{c}_{1\bar y}-\hat{c}_{3\bar y}}\,,  \\
\abs{\hat{c}_{2\tilde y}-\hat{c}_{3\tilde y}} \leq  2\pi - \abs{\hat{c}_{1\bar y'} + \hat{c}_{3\bar y'} - \pi} - \abs{\hat{c}_{1y'_1} + \hat{c}_{2y'_1} - \pi}\,, \\
\abs{\hat{c}_{1\bar y}-\hat{c}_{3\bar y}} + \abs{\hat{c}_{1y'} + \hat{c}_{2y'} - \pi} -\pi \leq \pi - \abs{\hat{c}_{2\tilde y'} + \hat{c}_{3\tilde y'} - \pi}\,,\\
\abs{\hat{c}_{1\bar y'} + \hat{c}_{3\bar y'} - \pi} + \abs{\hat{c}_{1y_1}-\hat{c}_{2y_1}} -\pi \leq \pi - \abs{\hat{c}_{2\tilde y'} + \hat{c}_{3\tilde y'} - \pi}.
\end{gather*}
This completes the proof.
\end{proof}

The above characterization give rise to a system of inequality that bounds the quantum set in all possible directions. We extract from this system the known inequalities and several ``apparently new" inequalities.
\begin{lmm}
If a correlation is quantum, i.e. $(c_{xy})\in\cor(3,3)$, then it must satisfy the box inequalities, the TLM-type inequalities
\begin{align*}
\abs{\hat{c}_{xy}-\hat{c}_{x'y}} + \abs{\hat{c}_{xy'} + \hat{c}_{x'y'} - \pi} \leq \pi
\end{align*}
for $x\neq x',y\neq y'$ and 
\begin{gather*}
\abs{\hat{c}_{1y}-\hat{c}_{2y}} + \abs{\hat{c}_{2y'}-\hat{c}_{3y'}} + \abs{\hat{c}_{1\bar y} -\hat{c}_{3\bar y}} \leq 2\pi\,,  \\
\abs{\hat{c}_{1y}-\hat{c}_{2y}} + \abs{\hat{c}_{2y'} + \hat{c}_{3y'} - \pi} + \abs{\hat{c}_{1\bar y} + \hat{c}_{3\bar y} - \pi}  \leq 2\pi\,, \\
\abs{\hat{c}_{1y} + \hat{c}_{2y} - \pi} + \abs{\hat{c}_{2y'}-\hat{c}_{3y'}} + \abs{\hat{c}_{1\bar y} + \hat{c}_{3\bar y} - \pi}  \leq 2\pi\,,  \\
\abs{\hat{c}_{1y} + \hat{c}_{2y} - \pi} + \abs{\hat{c}_{2y'} + \hat{c}_{3y'} - \pi} + \abs{\hat{c}_{1\bar y}-\hat{c}_{3\bar y}} \leq 2\pi\,,
\end{gather*}
for all $y,y',\bar{y}$ except $y=y'=\bar{y}$. In particular, if any of these inequalities fail to hold, then the correlation is not quantum. Moreover, these are all tight quantum Bell inequalities.
\end{lmm}
\begin{proof}
It is clear that the system of inequality in this Theorem is equivalent to that of Lemma~\ref{lmm:relax_equiv} (e.g. inequalities where $y=y'=\bar{y}$ follow from box inequalities). If any of these inequalities is not satisfied, then the correlation is not in the relaxation of $\cor(3,3)$ and hence not a quantum correlation. It remains to verify that these are all tight inequalities, i.e. there exists a quantum correlation that saturates them (individually). For simplicity, we check this numerically, but it is also possible to prove using chordal completion arguments (for a different graph depending on which inequality is tight).



\end{proof}

One might be tempted to conclude that, for instance, the inequalities
\begin{align*}
\abs{\hat{c}_{1y}-\hat{c}_{2y}} + \abs{\hat{c}_{2y'}-\hat{c}_{3y'}} + \abs{\hat{c}_{1\bar y} -\hat{c}_{3\bar y}} \leq 2\pi
\end{align*}
with $y\neq y'\neq \bar y$ are ``new" quantum Bell inequalities in the sense that it bounds the quantum correlation set and is not of the TLM type. However, the following Lemma shows that this is not really the case.

\begin{lmm}
The system of linear inequalities in Lemma~\ref{lmm:relax_equiv} is equivalent to the system of inequalities
\begin{align*}
\abs{\hat{c}_{xy}-\hat{c}_{x'y}} + \abs{\hat{c}_{xy'} + \hat{c}_{x'y'} - \pi} \leq \pi
\end{align*}
for all $x,y,x',y'$.
\end{lmm}
\begin{proof}
Let $P_1$ be the polytope defined by the inequalities in Lemma~\ref{lmm:relax_equiv} and $P_2$ be the polytope defined by the inequalities in this Lemma. We show that the extreme points or vertices of two polytopes are the same, hence the two system of inequalities are equivalent. The computation is done using MPT3~\cite{mpt3} and YALMIP~\cite{yalmip} packages in MATLAB.
\end{proof}

In summary, in the nonlinear coordinate system the linear relaxation of $\cor(3,3)$ is fully characterized by the box inequalities and the TLM-type inequalities.

\section{Generalizations}
Let us remark the bigger picture behind our method. Recall that $\cor(n,m)=\calE(K_{n,m})$
where $K_{n,m}$ is the complete bipartite graph on $n+m$ vertices and $\calE(K_{n,m})$ is the coordinate projection of the elliptope of $(n+m)\times(n+m)$ positive semidefinite matrices with diagonal one~\cite{Thinh}. Since for any graph $G$
\begin{align*}
\cut^{\pm1}(G)\subseteq\calE(G)\subseteq \cos(\pi\cut^{01}(G))\subseteq \cos(\pi\met^{01}(G))
\end{align*}
where $\cut^{\pm1}(G)$, $\cut^{01}(G)$, $\met^{01}(G)$ are {\em polytopes} induced by $G$ and cosine mapping is applied component-wise, the connection with $\calE(K_{n,m})$ allows one to compute approximations to the quantum correlation set via the middle inclusion~\footnote{Note that all inclusions are strict for the graph $K_{4,4}$.}. Moreover, the leftmost inclusion is an equality iff $G$ is acyclic, the rightmost inclusion is an equality iff $G$ has no $K_5$ minor, and the middle inclusion is an equality iff $G$ has no $K_4$ minor (see~\cite{L97} for more details).

Therefore, depending on the interested scenario, which gives rise to certain pattern of entries as captured by the graph $G$, one can get quantum Bell inequalities constraining the correlations by computing facets of the cut polytope $\cut^{01}(G)$. This is what we computed in Section~\ref{sec:relaxation}.


\section{Conclusions}
Explicit description of the quantum set $\cor(n,m)$ is expected to be extremely complicated for large $n,m$ because the geometric object is of very high dimension. We discovered that the complexity is already high for relatively small scenario $n=m=3$. We obtain exact characterization in terms of semidefinite programs, and exact but implicit characterizations in the linear and pairwise-angle coordinate systems. In the derivation, we observe the emergence of the tradeoff between implicit, compact description and explicit, complex description due to our use of Sylvester's criterion for positive semidefiniteness. (We do not know if this is intrinsic or one could hope to bypass Sylvester's criterion with a different characterization of positive-semidefiniteness.)

Fortunately, for the purpose of computing quantum Bell inequalities, a relaxation of the quantum correlation set suffices: $\cor(n,m)\subseteq \cos(\pi\cut^{01}(K_{n,m}))$. Thus, we are able to compute quantum Bell inequalities that bound the quantum set from the outside. We found that there are no new inequalities beyond the well-known Tsirelson-Landau-Masanes type. Finally, we sketch a general method to compute quantum Bell inequalities in more general scenarios including non-correlation scenarios. 

An obvious future direction is to compute an explicit description of $\cor(3,3)$ using either a modified quantifier elimination procedure for linear coordinates, or completely novel elimination techniques for pairwise-angle coordinates. In pairwise-angle coordinates, such result will shed light on the role of the 4-by-4 minor in constraining further the relaxation we obtained, and as well as the tightness of our relaxation. Another direction would be to investigate alternative parametrizations (and their geometrical or physical interpretations) that would linearize the description of the quantum set, analogous to the $\arccos$ parametrization linearizing $\cor(2,n)$.

\acknowledgements{We would like to thank Valerio Scarani and Antonios Varvitsiotis for helpful discussions. This work is supported by the National Research Fund and the Ministry of Education, Singapore, under the Research Centres of Excellence programme.}

\bibliography{qBellcites}

\begin{thebibliography}{25}%
\makeatletter
\providecommand \@ifxundefined [1]{%
 \@ifx{#1\undefined}
}%
\providecommand \@ifnum [1]{%
 \ifnum #1\expandafter \@firstoftwo
 \else \expandafter \@secondoftwo
 \fi
}%
\providecommand \@ifx [1]{%
 \ifx #1\expandafter \@firstoftwo
 \else \expandafter \@secondoftwo
 \fi
}%
\providecommand \natexlab [1]{#1}%
\providecommand \enquote  [1]{``#1''}%
\providecommand \bibnamefont  [1]{#1}%
\providecommand \bibfnamefont [1]{#1}%
\providecommand \citenamefont [1]{#1}%
\providecommand \href@noop [0]{\@secondoftwo}%
\providecommand \href [0]{\begingroup \@sanitize@url \@href}%
\providecommand \@href[1]{\@@startlink{#1}\@@href}%
\providecommand \@@href[1]{\endgroup#1\@@endlink}%
\providecommand \@sanitize@url [0]{\catcode `\\12\catcode `\$12\catcode
  `\&12\catcode `\#12\catcode `\^12\catcode `\_12\catcode `\%12\relax}%
\providecommand \@@startlink[1]{}%
\providecommand \@@endlink[0]{}%
\providecommand \url  [0]{\begingroup\@sanitize@url \@url }%
\providecommand \@url [1]{\endgroup\@href {#1}{\urlprefix }}%
\providecommand \urlprefix  [0]{URL }%
\providecommand \Eprint [0]{\href }%
\providecommand \doibase [0]{http://dx.doi.org/}%
\providecommand \selectlanguage [0]{\@gobble}%
\providecommand \bibinfo  [0]{\@secondoftwo}%
\providecommand \bibfield  [0]{\@secondoftwo}%
\providecommand \translation [1]{[#1]}%
\providecommand \BibitemOpen [0]{}%
\providecommand \bibitemStop [0]{}%
\providecommand \bibitemNoStop [0]{.\EOS\space}%
\providecommand \EOS [0]{\spacefactor3000\relax}%
\providecommand \BibitemShut  [1]{\csname bibitem#1\endcsname}%
\let\auto@bib@innerbib\@empty
\bibitem [{\citenamefont {Bell}(1964)}]{bell}%
  \BibitemOpen
  \bibfield  {author} {\bibinfo {author} {\bibfnamefont {J.~S.}\ \bibnamefont
  {Bell}},\ }\href@noop {} {\bibfield  {journal} {\bibinfo  {journal}
  {Physics}\ }\textbf {\bibinfo {volume} {1}},\ \bibinfo {pages} {195}
  (\bibinfo {year} {1964})}\BibitemShut {NoStop}%
\bibitem [{\citenamefont {Scarani}(2019)}]{bookNL}%
  \BibitemOpen
  \bibfield  {author} {\bibinfo {author} {\bibfnamefont {V.}~\bibnamefont
  {Scarani}},\ }\href@noop {} {\emph {\bibinfo {title} {Bell Nonlocality}}}\
  (\bibinfo  {publisher} {Oxford University Press},\ \bibinfo {year}
  {2019})\BibitemShut {NoStop}%
\bibitem [{\citenamefont {Brunner}\ \emph {et~al.}(2014)\citenamefont
  {Brunner}, \citenamefont {Cavalcanti}, \citenamefont {Pironio}, \citenamefont
  {Scarani},\ and\ \citenamefont {Wehner}}]{reviewNL}%
  \BibitemOpen
  \bibfield  {author} {\bibinfo {author} {\bibfnamefont {N.}~\bibnamefont
  {Brunner}}, \bibinfo {author} {\bibfnamefont {D.}~\bibnamefont {Cavalcanti}},
  \bibinfo {author} {\bibfnamefont {S.}~\bibnamefont {Pironio}}, \bibinfo
  {author} {\bibfnamefont {V.}~\bibnamefont {Scarani}}, \ and\ \bibinfo
  {author} {\bibfnamefont {S.}~\bibnamefont {Wehner}},\ }\href {\doibase
  10.1103/RevModPhys.86.419} {\bibfield  {journal} {\bibinfo  {journal} {Rev.
  Mod. Phys.}\ }\textbf {\bibinfo {volume} {86}},\ \bibinfo {pages} {419}
  (\bibinfo {year} {2014})}\BibitemShut {NoStop}%
\bibitem [{\citenamefont {Fine}(1982)}]{Fine}%
  \BibitemOpen
  \bibfield  {author} {\bibinfo {author} {\bibfnamefont {A.}~\bibnamefont
  {Fine}},\ }\href {\doibase 10.1103/PhysRevLett.48.291} {\bibfield  {journal}
  {\bibinfo  {journal} {Phys. Rev. Lett.}\ }\textbf {\bibinfo {volume} {48}},\
  \bibinfo {pages} {291} (\bibinfo {year} {1982})}\BibitemShut {NoStop}%
\bibitem [{\citenamefont {Werner}(1999)}]{allBell}%
  \BibitemOpen
  \bibfield  {author} {\bibinfo {author} {\bibfnamefont {R.}~\bibnamefont
  {Werner}},\ }\href
  {https://oqp.iqoqi.univie.ac.at/all-the-bell-inequalities/} {\enquote
  {\bibinfo {title} {All the bell inequalities},}\ } (\bibinfo {year}
  {1999})\BibitemShut {NoStop}%
\bibitem [{\citenamefont {Tsirelson}(1987)}]{TS87}%
  \BibitemOpen
  \bibfield  {author} {\bibinfo {author} {\bibfnamefont {B.~S.}\ \bibnamefont
  {Tsirelson}},\ }\href@noop {} {\bibfield  {journal} {\bibinfo  {journal}
  {Journal of Soviet Mathematics}\ }\textbf {\bibinfo {volume} {36}},\ \bibinfo
  {pages} {557} (\bibinfo {year} {1987})}\BibitemShut {NoStop}%
\bibitem [{\citenamefont {Landau}(1988)}]{Landau}%
  \BibitemOpen
  \bibfield  {author} {\bibinfo {author} {\bibfnamefont {L.~J.}\ \bibnamefont
  {Landau}},\ }\href@noop {} {\bibfield  {journal} {\bibinfo  {journal}
  {Foundations of Physics}\ }\textbf {\bibinfo {volume} {18}},\ \bibinfo
  {pages} {449} (\bibinfo {year} {1988})}\BibitemShut {NoStop}%
\bibitem [{\citenamefont {Masanes}(2005)}]{Masanes}%
  \BibitemOpen
  \bibfield  {author} {\bibinfo {author} {\bibfnamefont {L.}~\bibnamefont
  {Masanes}},\ }\href@noop {} {\bibfield  {journal} {\bibinfo  {journal} {arXiv
  preprint quant-ph/0512100}\ } (\bibinfo {year} {2005})}\BibitemShut {NoStop}%
\bibitem [{\citenamefont {Navascu{\'e}s}\ \emph {et~al.}(2008)\citenamefont
  {Navascu{\'e}s}, \citenamefont {Pironio},\ and\ \citenamefont
  {Ac{\'\i}n}}]{npa}%
  \BibitemOpen
  \bibfield  {author} {\bibinfo {author} {\bibfnamefont {M.}~\bibnamefont
  {Navascu{\'e}s}}, \bibinfo {author} {\bibfnamefont {S.}~\bibnamefont
  {Pironio}}, \ and\ \bibinfo {author} {\bibfnamefont {A.}~\bibnamefont
  {Ac{\'\i}n}},\ }\href@noop {} {\bibfield  {journal} {\bibinfo  {journal} {New
  Journal of Physics}\ }\textbf {\bibinfo {volume} {10}},\ \bibinfo {pages}
  {073013} (\bibinfo {year} {2008})}\BibitemShut {NoStop}%
\bibitem [{\citenamefont {Goh}\ \emph {et~al.}(2018)\citenamefont {Goh},
  \citenamefont {Kaniewski}, \citenamefont {Wolfe}, \citenamefont
  {V{\'e}rtesi}, \citenamefont {Wu}, \citenamefont {Cai}, \citenamefont
  {Liang},\ and\ \citenamefont {Scarani}}]{geometry}%
  \BibitemOpen
  \bibfield  {author} {\bibinfo {author} {\bibfnamefont {K.~T.}\ \bibnamefont
  {Goh}}, \bibinfo {author} {\bibfnamefont {J.}~\bibnamefont {Kaniewski}},
  \bibinfo {author} {\bibfnamefont {E.}~\bibnamefont {Wolfe}}, \bibinfo
  {author} {\bibfnamefont {T.}~\bibnamefont {V{\'e}rtesi}}, \bibinfo {author}
  {\bibfnamefont {X.}~\bibnamefont {Wu}}, \bibinfo {author} {\bibfnamefont
  {Y.}~\bibnamefont {Cai}}, \bibinfo {author} {\bibfnamefont {Y.-C.}\
  \bibnamefont {Liang}}, \ and\ \bibinfo {author} {\bibfnamefont
  {V.}~\bibnamefont {Scarani}},\ }\href@noop {} {\bibfield  {journal} {\bibinfo
   {journal} {Physical Review A}\ }\textbf {\bibinfo {volume} {97}},\ \bibinfo
  {pages} {022104} (\bibinfo {year} {2018})}\BibitemShut {NoStop}%
\bibitem [{\citenamefont {Mayers}\ and\ \citenamefont {Yao}(2003)}]{MY}%
  \BibitemOpen
  \bibfield  {author} {\bibinfo {author} {\bibfnamefont {D.}~\bibnamefont
  {Mayers}}\ and\ \bibinfo {author} {\bibfnamefont {A.}~\bibnamefont {Yao}},\
  }\href@noop {} {\bibfield  {journal} {\bibinfo  {journal} {arXiv preprint
  quant-ph/0307205}\ } (\bibinfo {year} {2003})}\BibitemShut {NoStop}%
\bibitem [{\citenamefont {Blume-Kohout}\ \emph {et~al.}()\citenamefont
  {Blume-Kohout}, \citenamefont {Gamble}, \citenamefont {Nielsen},
  \citenamefont {Mizrahi}, \citenamefont {Sterk},\ and\ \citenamefont
  {Maunz}}]{GST}%
  \BibitemOpen
  \bibfield  {author} {\bibinfo {author} {\bibfnamefont {R.~J.}\ \bibnamefont
  {Blume-Kohout}}, \bibinfo {author} {\bibfnamefont {J.~K.}\ \bibnamefont
  {Gamble}}, \bibinfo {author} {\bibfnamefont {E.}~\bibnamefont {Nielsen}},
  \bibinfo {author} {\bibfnamefont {J.~A.}\ \bibnamefont {Mizrahi}}, \bibinfo
  {author} {\bibfnamefont {J.~D.}\ \bibnamefont {Sterk}}, \ and\ \bibinfo
  {author} {\bibfnamefont {P.~L.~W.}\ \bibnamefont {Maunz}},\ }\href@noop {}
  {\bibinfo  {journal} {Nature Physics}\ }\BibitemShut {NoStop}%
\bibitem [{\citenamefont {Yang}\ \emph {et~al.}(2011)\citenamefont {Yang},
  \citenamefont {Navascu\'es}, \citenamefont {Sheridan},\ and\ \citenamefont
  {Scarani}}]{qBell_ML}%
  \BibitemOpen
\bibfield  {journal} {  }\bibfield  {author} {\bibinfo {author} {\bibfnamefont
  {T.~H.}\ \bibnamefont {Yang}}, \bibinfo {author} {\bibfnamefont
  {M.}~\bibnamefont {Navascu\'es}}, \bibinfo {author} {\bibfnamefont
  {L.}~\bibnamefont {Sheridan}}, \ and\ \bibinfo {author} {\bibfnamefont
  {V.}~\bibnamefont {Scarani}},\ }\href {\doibase 10.1103/PhysRevA.83.022105}
  {\bibfield  {journal} {\bibinfo  {journal} {Phys. Rev. A}\ }\textbf {\bibinfo
  {volume} {83}},\ \bibinfo {pages} {022105} (\bibinfo {year}
  {2011})}\BibitemShut {NoStop}%
\bibitem [{\citenamefont {Thinh}\ \emph {et~al.}(2019)\citenamefont {Thinh},
  \citenamefont {Varvitsiotis},\ and\ \citenamefont {Cai}}]{Thinh}%
  \BibitemOpen
  \bibfield  {author} {\bibinfo {author} {\bibfnamefont {L.~P.}\ \bibnamefont
  {Thinh}}, \bibinfo {author} {\bibfnamefont {A.}~\bibnamefont {Varvitsiotis}},
  \ and\ \bibinfo {author} {\bibfnamefont {Y.}~\bibnamefont {Cai}},\ }\href
  {\doibase 10.1103/PhysRevA.99.052108} {\bibfield  {journal} {\bibinfo
  {journal} {Phys. Rev. A}\ }\textbf {\bibinfo {volume} {99}},\ \bibinfo
  {pages} {052108} (\bibinfo {year} {2019})}\BibitemShut {NoStop}%
\bibitem [{\citenamefont {Grone}\ \emph {et~al.}(1984)\citenamefont {Grone},
  \citenamefont {Johnson}, \citenamefont {S\'a},\ and\ \citenamefont
  {Wolkowicz}}]{wolk}%
  \BibitemOpen
  \bibfield  {author} {\bibinfo {author} {\bibfnamefont {R.}~\bibnamefont
  {Grone}}, \bibinfo {author} {\bibfnamefont {C.}~\bibnamefont {Johnson}},
  \bibinfo {author} {\bibfnamefont {E.}~\bibnamefont {S\'a}}, \ and\ \bibinfo
  {author} {\bibfnamefont {H.}~\bibnamefont {Wolkowicz}},\ }\href {\doibase
  https://doi.org/10.1016/0024-3795(84)90207-6} {\bibfield  {journal} {\bibinfo
   {journal} {Linear Algebra and its Applications}\ }\textbf {\bibinfo {volume}
  {58}},\ \bibinfo {pages} {109 } (\bibinfo {year} {1984})}\BibitemShut
  {NoStop}%
\bibitem [{\citenamefont {Laurent}(1997)}]{L97}%
  \BibitemOpen
  \bibfield  {author} {\bibinfo {author} {\bibfnamefont {M.}~\bibnamefont
  {Laurent}},\ }\href@noop {} {\bibfield  {journal} {\bibinfo  {journal}
  {Linear Algebra and its Applications}\ }\textbf {\bibinfo {volume} {252}},\
  \bibinfo {pages} {347} (\bibinfo {year} {1997})}\BibitemShut {NoStop}%
\bibitem [{\citenamefont {Barrett}\ \emph {et~al.}(1993)\citenamefont
  {Barrett}, \citenamefont {Johnson},\ and\ \citenamefont {Tarazaga}}]{BJT}%
  \BibitemOpen
  \bibfield  {author} {\bibinfo {author} {\bibfnamefont {W.}~\bibnamefont
  {Barrett}}, \bibinfo {author} {\bibfnamefont {C.~R.}\ \bibnamefont
  {Johnson}}, \ and\ \bibinfo {author} {\bibfnamefont {P.}~\bibnamefont
  {Tarazaga}},\ }\href@noop {} {\bibfield  {journal} {\bibinfo  {journal}
  {Linear Algebra and its Applications}\ }\textbf {\bibinfo {volume} {192}},\
  \bibinfo {pages} {3} (\bibinfo {year} {1993})}\BibitemShut {NoStop}%
\bibitem [{\citenamefont {Tanigawa}(2017)}]{shin}%
  \BibitemOpen
  \bibfield  {author} {\bibinfo {author} {\bibfnamefont {S.-I.}\ \bibnamefont
  {Tanigawa}},\ }\href@noop {} {\bibfield  {journal} {\bibinfo  {journal} {SIAM
  Journal on Optimization}\ }\textbf {\bibinfo {volume} {27}},\ \bibinfo
  {pages} {986} (\bibinfo {year} {2017})}\BibitemShut {NoStop}%
\bibitem [{\citenamefont {Horn}\ and\ \citenamefont
  {Johnson}(2012)}]{matrix_analysis}%
  \BibitemOpen
  \bibfield  {author} {\bibinfo {author} {\bibfnamefont {R.~A.}\ \bibnamefont
  {Horn}}\ and\ \bibinfo {author} {\bibfnamefont {C.~R.}\ \bibnamefont
  {Johnson}},\ }\href@noop {} {\emph {\bibinfo {title} {Matrix Analysis}}},\
  \bibinfo {edition} {2nd}\ ed.\ (\bibinfo  {publisher} {Cambridge University
  Press},\ \bibinfo {address} {New York, NY, USA},\ \bibinfo {year}
  {2012})\BibitemShut {NoStop}%
\bibitem [{\citenamefont {Caviness}\ and\ \citenamefont
  {Johnson}(2012)}]{QECAD}%
  \BibitemOpen
  \bibfield  {author} {\bibinfo {author} {\bibfnamefont {B.~F.}\ \bibnamefont
  {Caviness}}\ and\ \bibinfo {author} {\bibfnamefont {J.~R.}\ \bibnamefont
  {Johnson}},\ }\href@noop {} {\emph {\bibinfo {title} {Quantifier elimination
  and cylindrical algebraic decomposition}}}\ (\bibinfo  {publisher} {Springer
  Science \& Business Media},\ \bibinfo {year} {2012})\BibitemShut {NoStop}%
\bibitem [{\citenamefont {Brown}(2003)}]{QEPCAD}%
  \BibitemOpen
  \bibfield  {author} {\bibinfo {author} {\bibfnamefont {C.~W.}\ \bibnamefont
  {Brown}},\ }\href {\url{http://www.cs.usna.edu/~qepcad/B/QEPCAD.html}}
  {\enquote {\bibinfo {title} {{QEPCAD Quantifier Elimination by Partial
  Cylindrical Algebraic Decomposition}},}\ } (\bibinfo {year}
  {2003})\BibitemShut {NoStop}%
\bibitem [{\citenamefont {Sturm}\ and\ \citenamefont
  {Dolzmann}(2015)}]{REDLOG}%
  \BibitemOpen
  \bibfield  {author} {\bibinfo {author} {\bibfnamefont {T.}~\bibnamefont
  {Sturm}}\ and\ \bibinfo {author} {\bibfnamefont {A.}~\bibnamefont
  {Dolzmann}},\ }\href {\url{http://www.redlog.eu}} {\enquote {\bibinfo {title}
  {{REDLOG Computing with Logic}},}\ } (\bibinfo {year} {2015})\BibitemShut
  {NoStop}%
\bibitem [{\citenamefont {Herceg}\ \emph {et~al.}(2013)\citenamefont {Herceg},
  \citenamefont {Kvasnica}, \citenamefont {Jones},\ and\ \citenamefont
  {Morari}}]{mpt3}%
  \BibitemOpen
  \bibfield  {author} {\bibinfo {author} {\bibfnamefont {M.}~\bibnamefont
  {Herceg}}, \bibinfo {author} {\bibfnamefont {M.}~\bibnamefont {Kvasnica}},
  \bibinfo {author} {\bibfnamefont {C.}~\bibnamefont {Jones}}, \ and\ \bibinfo
  {author} {\bibfnamefont {M.}~\bibnamefont {Morari}},\ }in\ \href@noop {}
  {\emph {\bibinfo {booktitle} {Proc.~of the European Control Conference}}}\
  (\bibinfo {address} {Z\"urich, Switzerland},\ \bibinfo {year} {2013})\ pp.\
  \bibinfo {pages} {502--510},\ \bibinfo {note}
  {\url{http://control.ee.ethz.ch/~mpt}}\BibitemShut {NoStop}%
\bibitem [{\citenamefont {L{\"{o}}fberg}(2004)}]{yalmip}%
  \BibitemOpen
  \bibfield  {author} {\bibinfo {author} {\bibfnamefont {J.}~\bibnamefont
  {L{\"{o}}fberg}},\ }in\ \href@noop {} {\emph {\bibinfo {booktitle} {In
  Proceedings of the CACSD Conference}}}\ (\bibinfo {address} {Taipei,
  Taiwan},\ \bibinfo {year} {2004})\BibitemShut {NoStop}%
\bibitem [{Note1()}]{Note1}%
  \BibitemOpen
  \bibinfo {note} {Note that all inclusions are strict for the graph
  $K_{4,4}$.}\BibitemShut {Stop}%
\end{thebibliography}%

\end{document}